%% file: main.tex
\documentclass[conference,letterpaper]{IEEEtran}
\usepackage[utf8]{inputenc} 
\usepackage[T1]{fontenc}
\usepackage{url}
\usepackage{ifthen}
\usepackage{cite}
\usepackage[cmex10]{amsmath} 
\usepackage{amsmath, amssymb, amsthm}
\usepackage{bm}
\usepackage{bbm}
\usepackage{color}
\usepackage{graphicx}
\usepackage{mathrsfs}
\usepackage{xcolor}
\usepackage{caption}
\usepackage{subcaption}

\input{./macro_cfli.tex}

\newtheorem{theorem}{Theorem}

\newtheorem{lemma}{Lemma}
\newtheorem{proposition}{Proposition}
\newtheorem{corollary}{Corollary}

\theoremstyle{definition}
\newtheorem{definition}{Definition}
\theoremstyle{remark}

\newtheorem{claim}{Claim}

\begin{document}
\title{Improved Bounds on Access-Redundancy Tradeoffs in Quantized Linear Computations}

\author{%
\IEEEauthorblockN{Ching-Fang Li}
\IEEEauthorblockA{
Department of Electrical Engineering\\
Stanford University, CA, USA\\
Email: \url{cfli@stanford.edu}
}
\and
\IEEEauthorblockN{Mary Wootters}
\IEEEauthorblockA{
Departments of Computer Science and Electrical Engineering\\
Stanford University, CA, USA\\
Email: \url{marykw@stanford.edu}
}
}

\maketitle

\begin{abstract}
\input{v2/abstract.tex}
\end{abstract}

\section{Introduction}\label{sec:intro}
\input{v2/intro.tex}

\section{Preliminaries}\label{sec:formulation}
\input{v2/formulation.tex}

\section{Main Results}\label{sec:lowerbound}
\input{v2/lower_bound.tex}

\section{Approximate Recovery}\label{sec:approx}
\input{v2/approx.tex}

\section*{Acknowledgments}
M.W. is partially funded by NSF grants CCF-2231157 and CNS-2321489.

\IEEEtriggeratref{6}
\bibliographystyle{IEEEtran}
\bibliography{./bib/Ref.bib}

\clearpage
\onecolumn


\end{document}

%% file: macro_cfli.tex
\newcommand{\lbp}{\left\{}
\newcommand{\rbp}{\right\}}

\newcommand{\lnorm}{\left\lVert}
\newcommand{\rnorm}{\right\rVert}
\newcommand{\mv}{\,\middle\vert\,}

\newcommand{\mcal}{\mathcal}
\newcommand{\mscr}{\mathscr}

\newcommand{\mb}{\mathbf}
\newcommand{\mbb}{\mathbb}

\newcommand{\lfl}{\left\lfloor}
\newcommand{\rfl}{\right\rfloor}

\newcommand{\E}{\mathbb{E}}

\renewcommand{\Pr}{\mathbb{P}}

%% file: v2/abstract.tex
Consider the problem of computing quantized linear functions with only a few queries.  Formally, given $\mb x\in \mbb R^k$, our goal is to encode $\mb x$ as $\mb c \in \mbb R^n$, for $n > k$, so that for any $\mb w \in \mathcal{A}^k$, $\mb w^T \mb x$ can be computed using at most $\ell$ queries to $\mb c$.  Here, $\mathcal{A}$ is some finite set; in this paper we focus on the case where $|\mathcal{A}| = 2$.  

Prior work \emph{(Ramkumar, Raviv, and Tamo, Trans. IT, 2024)} has given constructions and established impossibility results for this problem.  We give improved impossibility results, both for the general problem, and for the specific class of construction (block construction) presented in that work.  The latter establishes that the block constructions of prior work are optimal within that class.

We also initiate the study of \emph{approximate} recovery for this problem, where the goal is not to recover $\mb w^T \mb x$ exactly but rather to approximate it up to a parameter $\varepsilon > 0$.  We give several constructions, and give constructions for $\varepsilon = 0.1$ that outperform our impossibility result for exact schemes.

%% file: v2/intro.tex
In this paper we consider the problem of computing quantized linear functions of data with a few queries.  Formally, consider the set of linear functions $\mathcal{F}_{\mathcal{A}} = \{f_{\mb w} \mid \mb w \in \mathcal{A}^k\}$ for some finite set $\mathcal{A}$, where $f_{\mb w}:{\mbb R}^k \to \mbb R$ is given by $f_{\mb w}(\mb x) = {\mb w}^T{\mb x}.$  The goal is to encode $\mb x \in \mbb R^k$ into some $\mb y \in \mbb R^n$, for $n \geq k$, so that, for all $f \in \mathcal{F}$, it is possible to recover $f(\mb x)$ given only $\ell \leq n$ queries to the symbols of $\mb y$.  
The parameters of interest are the \emph{redundancy ratio} $\nu = n/k$, which captures the amount of overhead in the encoding; and the \emph{access ratio} $\lambda = \ell/k$, which captures the query complexity.  We would like both to be as small as possible.  If there is an infinite family of coding schemes that achieves the parameters $(\nu, \lambda)$, we say that this pair is \emph{feasible} (see Definition~\ref{def:feasible}).
Understanding the set of feasible pairs, and coming up with code constructions that achieve them, is motivated by applications in distributed ML and coded computation. (See the discussion in \cite{ramkumar_access-redundancy_2023, RRT23_coeff_complexity, RRT24_nonbinary}. Briefly, we imagine that $\mb x$ is stored by storing the symbols of $\mb y$ each at different nodes in a distributed system.  Our goal is to compute $f(\mb x)$ while accessing as few nodes as possible).

This problem has been studied for various families $\mathcal{F}$ over various domains.  If $\mb x \in \mbb F^k$ for some finite field $\mbb F$, and if $\mathcal{A} = \mbb F$, then the problem is equivalent to the existence of linear covering codes (see \cite{cohen1997covering} and the discussion in \cite{ramkumar_access-redundancy_2023}).  Over $\mbb R$, \cite{HoAgrawal_98} showed how to use covering codes (in a different way) to obtain nontrivial constructions for $\mathcal{A} = \{0,1\}$.  This was later improved by \cite{ramkumar_access-redundancy_2023} for any $\mathcal{A} \subseteq \mbb R^k$ with $|\mathcal{A}| = 2$: That work showed that the problem is equivalent for any $\mathcal{A}$ with $|\mathcal{A}| = 2$, and provided both improved constructions and impossibility results.  In further work, \cite{RRT23_coeff_complexity, RRT24_nonbinary} investigated the setting where $|\mathcal{A}| > 2$.

\textbf{Our Contributions.}
In this paper, we work in the setting of \cite{HoAgrawal_98, ramkumar_access-redundancy_2023}, where $|\mathcal{A}| = 2$.  As in prior work, we focus on schemes with linear encoding and decoding maps. Our main contributions are the following.
\begin{itemize}
    \item \textbf{Improved Lower Bounds on $\nu, \lambda$.}
    In \cite{ramkumar_access-redundancy_2023}, an impossibility result is proved on the set of feasible points $(\nu, \lambda).$ Using \emph{covering designs}, we are able to improve this bound (Theorem~\ref{thm:lower_bound}; see also Corollary~\ref{cor:main}).  A comparison of the our new bound with the bound from prior work is shown in Figure~\ref{fig:compare}.

    We further prove specialized lower bounds for the class of constructions in \cite{ramkumar_access-redundancy_2023}.  That work presents a class of \emph{block constructions}, described more below, where $\mb x \in \mbb R^k$ is broken up into $m$ blocks in $\mbb R^{k_0}$, and each block is encoded separately.  We prove a lower bound for any block construction (Theorem~\ref{thm:block}), and we show that the constructions presented in \cite{ramkumar_access-redundancy_2023} in fact meet these bounds.  We also present a new block construction (given in \eqref{eq:new_construction} and also shown in Figure~\ref{fig:compare}) to demonstrate that other block constructions on our lower bound are achievable.

    \item \textbf{Initiating the study of approximate recovery.}  Prior work has focused on recovering $f(\bm x)$ exactly for all $f \in \mathcal{F}$.  In this work, we consider a relaxation (again for quantized linear functions with $|\mathcal{A}|=2$), where it is sufficient to estimate $f(\mb x)$ up to error $\varepsilon$.  We give three types of constructions for such approximation schemes, and show that even for $\varepsilon = 0.1$, it is possible to outperform our improved lower bound for exact reconstruction.  Our constructions are shown in Figure~\ref{fig:approx}.
\end{itemize}

\textbf{Outline.} In Section~\ref{sec:formulation}, we set notation and formally define the problem, as well as formally present a few of the results in \cite{ramkumar_access-redundancy_2023} that will be relevant for our work.  In Section~\ref{sec:lowerbound}, we present our improved general lower bound (Theorem~\ref{thm:lower_bound}), as well as our lower bound for block constructions (Theorem~\ref{thm:block}).  In Section~\ref{sec:approx}, we set up the approximate recovery problem, and present our constructions.

%% file: v2/formulation.tex
All vectors are column vectors. 
For a vector $\mb x\in \mbb R^k$, the entries are written as $(x_1,x_2,\dots,x_k)$. 
We use $\lnorm\mb x\rnorm_0$ to denote the number of nonzero entries of $\mb x$. 
We let $[n]=\{1,2,\dots,n\}$. For a finite set $\mcal T$ and an integer $0\leq \ell\leq |\mcal T|$, $\binom{\mcal T}{\ell}$ refers to the collection of all size-$\ell$ subsets of $\mcal T$. For a set $\mcal S \subseteq [n]$ and a matrix $\mb D$ with $n$ columns, $\mb D|_{\mcal S}$ denotes the set of columns of $\mb D$ indexed by $\mcal S$.  Logarithms are of base $2$ if not otherwise specified.  

\subsection{Problem Statement}
As stated in the introduction, our goal is to design an encoding map that maps $\mb x \in \mbb R^k$ to $\mb y \in \mbb R^n$. 
We will restrict to linear codes, meaning that the encoding map is given by $\mb y = \mb D^T \mb x$ 
for some matrix $\mb D \in \mbb R^{k \times n}$.
For any function $f:\mbb R^k \to \mbb R$ in some function class $\mathcal{F}$, we would like to be able to recover $f(\mb x)$ given at most $\ell$ queries to the symbols of $\mb y$. We focus on quantized linear computations by considering the function class $\mcal F_{\mcal A}=\lbp f_{\mb w} \mv \mb w\in\mcal A^k \rbp$ where $f_{\mb w}(\mb x) = \mb w^T \mb x$,
for some finite coefficient set $\mcal A\subset\mbb R$.
We also restrict our attention to linear decoding maps.  That is, for each $f \in \mcal F_{\mcal A},$ we want that $f(\mb x)$ can be represented as some linear combination of 
at most $\ell$ entries of $\mb y.$
Since everything we wish to compute is linear, we may rephrase the problem in terms of linear algebra.  Formally, we define a $\mcal A$-protocol as follows.

\begin{definition}[$\mcal A$-protocol]\label{def:protocol}
Fix a finite set $\mcal A \subseteq \mbb R$, and let $s = |\mcal A|^k$.  Let $\mb W \in \mbb R^{k \times s}$ be a matrix that has each $\mb w \in \mcal A^k$ as a column. A \emph{$\mcal A$-protocol} with parameters $(k,n,\ell)$ consists of an encoding matrix $\mb D \in \mbb R^{k \times n}$ and a decoding matrix $\mb A \in \mbb R^{n \times s}$, so that $\mb W = \mb D \mb A,$ and so that $\|\mb a\|_0 \leq \ell$ for every column $\mb a$ of $\mb A$.
\end{definition}

To see why Definition~\ref{def:protocol} captures the problem above, note that any $\mb w \in \mcal A^k$ appears as a column of $\mb W$.  To compute $f_{\mb w}(\mb x)$, let $\mb a \in \mbb R^{n}$ be the corresponding column of $\mb A$.  Then $f_{\mb w}(\mb x) = \mb w^T \mb x = \mb a^T \mb D^T \mb x = \mb a^T \mb y,$ and $\mb a$ has at most $\ell$ nonzero entries.  

Given $k$, we want both $n$ and $\ell$ to be small. 
Following \cite{ramkumar_access-redundancy_2023}, we study the \emph{access ratio} $\ell/k$ and the \emph{redundancy ratio} $n/k$ in the asymptotic regime.
\begin{definition}[Feasible pair]\label{def:feasible}
    We say that a pair $(\nu,\lambda)\in\mbb R^2$ is $\mcal A$-feasible if there exists an infinite family of $\mcal A$-protocols with parameters $\{(k_i,n_i,\ell_i)\}_{i\geq 1}$ (with $k_i$’s strictly increasing) such that $\lim_{i\to\infty} (n_i/k_i, \ell_i/k_i) = (\nu,\lambda)$. 
\end{definition}
\noindent It is shown in \cite{ramkumar_access-redundancy_2023} that taking any $\mcal A$ with size two is equivalent in the sense that the resulting feasible pairs are the same. Hence for the rest of the paper, we take $\mcal A=\{\pm 1\}$.
The goal is to characterize the region of feasible pairs.
In other words, we aim to find the optimal tradeoff between the access and redundancy ratios.

\subsection{Prior Work}\label{sec:related}
As noted in the introduction, this and similar problems have been studied before~\cite{HoAgrawal_98, ramkumar_access-redundancy_2023,RRT23_coeff_complexity,RRT24_nonbinary}. 
 The work most relevant to ours is that of Ramkumar, Raviv and Tamo in~\cite{ramkumar_access-redundancy_2023}. Here, we briefly review the technical details of their results, since they will be relevant for ours.

\subsubsection{Block Construction of \cite{ramkumar_access-redundancy_2023}}\label{subsubsec:construct}
One of the constructions of \cite{ramkumar_access-redundancy_2023} is based on covering codes.
Given a covering code $\mcal C \subseteq \{\pm 1\}^{k_0}$ of length $k_0$ over alphabet $\{\pm 1\}$, 
the \emph{covering radius} is 
\[\textstyle r=r(\mcal C)=\min\lbp r'\mv \bigcup_{\mb c\in\mcal C}B(\mb c, r')=\{\pm 1\}^{k_0} \rbp,\]
where $B(\mb c, r')$ is the Hamming ball centered at $\mb c$ with radius $r'$.
Let $\hat{\mcal C}\subseteq\mcal C$ be a set that contains exactly one of $\pm \mb c$ 
for each $\mb c\in\mcal C$, and let $\hat c=|\hat{\mcal C}|$.

It is shown in \cite[Theorem 2]{ramkumar_access-redundancy_2023} that the pair $(\frac{k_0+\hat c}{k_0}, \frac{r+1}{k_0})$ is feasible. (Notice that if $\mcal C$ is closed under complement (i.e. $\mb c\in\mcal C$ if and only if $-\mb c\in\mcal C$), then $\hat c=|\mcal C|/2$.)  To show this, they use the following construction, which divides the data $\mb x$ into blocks.  
Given data $\mb x$ with length $k=mk_0$, the idea is to divide $\mb x$ into $m$ blocks $\mb x_1,\dots,\mb x_m$, each of length $k_0$. 
Let $\mb D=\mb I_m\otimes\mb M$, where $\mb M=(\mb I_{k_0}|\mb B)\in\mbb R^{k_0\times(k_0+\hat c)}$ and $\mb B$ contains all $\mb c\in\hat{\mcal C}$ 
as columns, then $n= m(k_0+\hat c)$.
For $\mb w\in\{\pm 1\}^k$, write $\mb w=(\mb w_1,\dots,\mb w_m)$ with $\mb w_i\in\{\pm 1\}^{k_0}$. For each $\mb w_i$, since $\mcal C$ is a covering code with radius $r$, there exists $\mb c\in\mcal C$ such that the Hamming distance between $\mb c$ and $\mb w_i$ is at most $r$. Hence $\mb w_i^{\mathrm T}\mb x_i$ can be computed by querying the node containing $\mb c^{\mathrm T}\mb x_i$ or $-\mb c^{\mathrm T}\mb x_i$ along with at most $r$ other nodes to correct the difference. As a result, $\mb w^{\mathrm T}\mb x=\sum_{i=1}^m \mb w_i^{\mathrm T}\mb x_i$ can computed by querying at most $\ell=m(r+1)$ nodes.
Moreover, if $\mcal C=\{\pm 1\}^{k_0}$, then $r=0$ and the pair $(\frac{2^{k_0-1}}{k_0}, \frac{1}{k_0})$ is feasible by simply taking $\mb{M=B}$. 
In the remainder of the paper, we will refer to this as the \emph{non-systematic} construction.\footnote{We note that taking $\mb M=(\mb I_{k_0}|\mb B)$ results in
systematic constructions, meaning that $\mb y$ contains $\mb x$ as a substring.  However, these will not be relevant for our results.}

\subsubsection{Lower Bound of \cite{ramkumar_access-redundancy_2023}}
The paper \cite{ramkumar_access-redundancy_2023} proves the following impossibility result.
\begin{theorem}[Lower Bound in \cite{ramkumar_access-redundancy_2023}]\label{thm:old_bound}
    A $\{\pm 1\}$-protocol with parameters $(k, n,\ell)$ must satisfy $\binom{n}{\ell}2^\ell\geq 2^k$.
\end{theorem}
To prove this result, \cite{ramkumar_access-redundancy_2023} make use of the fact that
any $\ell$-dimensional $\mbb R$-subspace contains at most $2^\ell$ $\{\pm 1\}$-vectors \cite[Lemma 1]{RavivYaakobi_22}. Suppose that $(\mb D, \mb A)$ form a $\{\pm 1\}$-protocol with parameters $(k,n,\ell)$.  For each $\mb w\in\{\pm 1\}^k$, a subset of the $n$ nodes of size $\ell$ is accessed to compute $\mb w^{\mathrm T}\mb x$.\footnote{Notice that we may assume without loss of generality that exactly $\ell$ nodes are accessed, rather than up to $\ell$.}
There are at most $\binom{n}{\ell}$ such sets.  For each such set $\mcal S \subseteq [n]$, consider the subspace spanned by $\mb D|_{\mcal S}$. 
Each of these sets spans an $\ell$-dimensional subspace of $\mbb R^k$, so by the fact above, each such subspace contains at most $2^\ell$ $\{\pm 1\}$-vectors.  Since the union of all of these spans must contain all $\{\pm 1\}$-vectors, this proves Theorem~\ref{thm:old_bound}.

A comparison between the constructions and the numerical evaluation of the lower bound can be found in \cite{ramkumar_access-redundancy_2023} as well as in Figure~\ref{fig:compare}. Notice that there is a substantial gap between the constructions and lower bound. In the next section, we present an improved lower bound, narrowing the gap.

%% file: v2/lower_bound.tex
\subsection{General Lower Bound}
Our first result is an improvement to Theorem~\ref{thm:old_bound}.  To get some intuition for how Theorem~\ref{thm:old_bound} might be improved, notice that
in the proof of Theorem~\ref{thm:old_bound}, the $\ell$-dimensional subspaces spanned by the sets $\mb D|_{\mcal S}$ may have nontrivial intersections.  Thus, bounding the size of their union by the sum of their sizes might be loose. 
To improve the bound, will focus on subsets $\mcal T$ of size $t \geq \ell$, and apply the above logic to $\mcal T$, using the fact that a single $\mcal T$ contains multiple sets $\mcal S$ of size $\ell$. 

To make this intuition precise, let $t\geq \ell$ (to be chosen later), and let $\mcal T \subseteq [n]$ be a subset of $t$ nodes. By the fact noted above, there are at most $2^t$ $\{\pm 1\}$-vectors in the span of $\mb D|_{\mcal T}$. As a result,
\begin{equation}
    \Bigg\vert \bigcup_{\mcal S\in\binom{\mcal T}{\ell}} \Big\{\{\pm 1\} \text{-vectors in the span of } \mb D|_{\mcal S} \Big\} \Bigg\vert\leq 2^t.\label{eq:key_observation}
\end{equation}
With $t$ carefully chosen, this gives a tighter bound than the union bound $\binom{t}{\ell}2^\ell$. To present our improved lower bound, we introduce the notion of a \emph{covering design}.

\begin{definition}[Covering Design]
    An $(n,t,\ell)$-covering design is a family of $t$-subsets of a given $n$-element set such that each $\ell$-subset is contained in at least one of the $t$-subsets. Let $C(n,t,\ell)$ be the minimum size of such family.
\end{definition}

We will use the following bound on the size of covering designs, from~\cite{erdos_probabilistic_1974}.
\begin{lemma}[Covering Design \cite{erdos_probabilistic_1974}]\label{lemma:covering_design} $C(n,t,\ell) \leq \big[1+\ln\binom{t}{\ell}\big]\frac{\binom{n}{\ell}}{\binom{t}{\ell}}$.
\end{lemma}

Lemma~\ref{lemma:covering_design} is proved using probabilistic method, where a randomly generated covering design is shown to have expected size $\big[1+\ln\binom{t}{\ell}\big]\binom{n}{\ell}/\binom{t}{\ell}$. By Lemma~\ref{lemma:covering_design} and the key observation \eqref{eq:key_observation}, we obtain our improved lower bound.

\begin{theorem}[General Lower Bound]\label{thm:lower_bound}\label{thm:main}
    A $\{\pm 1\}$-protocol with parameters $(k,n,\ell)$ must satisfy, for any $t \in \{\ell, \ell+1, \ldots, n\}$,
    \begin{equation}
        2^k \leq \big[1+\ln\textstyle\binom{t}{\ell}\big]\displaystyle\frac{\binom{n}{\ell}}{\binom{t}{\ell}}2^t.\label{eq:main}
    \end{equation}
    
\end{theorem}
\begin{proof}
Let $\mscr F$ be a minimal $(n,t,\ell)$-covering design with size $C(n,t,\ell)$. Then 
\begin{align*}
    2^k &\leq \Bigg\vert \bigcup_{\mcal S\in\binom{[n]}{\ell}} \Big\{\{\pm 1\} \text{-vectors in the span of } \mb D|_{\mcal S} \Big\} \Bigg\vert\\
    &\leq \sum_{\mcal T\in\mscr F} \Bigg\vert \bigcup_{\mcal S\in\binom{\mcal T}{\ell}} \Big\{\{\pm 1\} \text{-vectors in the span of } \mb D|_{\mcal S} \Big\} \Bigg\vert\\
    &\leq C(n,t,\ell) \cdot 2^t \leq \big[1+\ln\textstyle\binom{t}{\ell}\big]\displaystyle\frac{\binom{n}{\ell}}{\binom{t}{\ell}}2^t.
\end{align*}
Above, the first inequality follows from the same logic as in the proof of Theorem~\ref{thm:old_bound};
the second inequality follows from the fact that $\mscr F$ is a $(n,t,\ell)$-covering design; and the third inequality follows from \eqref{eq:key_observation}.
\end{proof}
Next we evaluate this bound in the asymptotic regime where $k$ goes to infinity and $n,\ell$ are linear in $k$. Denote $\nu=\frac{n}{k}$ and $\lambda=\frac{\ell}{k}$. Take $t=r\ell-O(1)$ for some constant $r\geq 1$. We will later show which specific $r$ to choose and why it is valid.
Theorem~\ref{thm:main} implies that
\begin{align}
    1&\leq \textstyle \left[\log(1+\ln\binom{t}{\ell}) + \log\binom{n}{\ell} - \log \binom{t}{\ell}+ t\right]/k\nonumber\\
    &= \nu H(\lambda/\nu) - r\lambda H(1/r) + r\lambda + o(1),\label{eq:asymptotic_bound}
\end{align}
where $H$ is the binary entropy function and the last equality follows by the well-known bound
\[\textstyle\frac{1}{n+1}2^{nH(\ell/n)}\leq \binom{n}{\ell}\leq 2^{nH(\ell/n)}.\]
Taking $k\to\infty$, the inequality \eqref{eq:asymptotic_bound} gives
\[H(\lambda/\nu)\geq \left[1-r\big(1-H(1/r)\big)\lambda\right]/\nu.\]
Notice that $r\big(1-H(1/r)\big)$ is minimized at $r=2$ with value 0. Thus, choosing $r=2$ gives the best bound. To see that this is a valid choice, we show that it is possible to take $t\leq n$ with $t=2\ell-O(1)$. This is true because we always have $k\leq n$, and we know from \cite{ramkumar_access-redundancy_2023,HoAgrawal_98} that any protocol of interest has $\ell\leq k/2 + O(1)$ (indeed, the ``trivial'' protocol\footnote{This construction, noted in \cite{HoAgrawal_98}, is to store each $x_j$ with one additional parity check $\sum_{j=1}^k x_j$. In this construction, one can compute $\mb w^T \mb x$ for any $\mb w \in \{\pm 1\}^k$ by accessing $\ell \leq k/2+1$ nodes---the parity check along with at most half of the systematic nodes, corresponding to all the positive (or negative) coefficients of $\mb w$.} achieves the pair $(1,1/2)$, and increasing the redundancy ratio should improve the access ratio
).  
\begin{corollary}\label{cor:main}
A $\{\pm 1\}$-feasible pair $(\nu,\lambda)$ with $\lambda\leq 1/2$ must satisfy $H(\lambda/\nu)\geq 1/\nu$.
\end{corollary}
A numerical evaluation of Corollary~\ref{cor:main} is given in Figure~\ref{fig:compare}, along with the constructions and lower bound in \cite{ramkumar_access-redundancy_2023} for comparison. It is clear that our new lower bound represents substantial progress in closing the gap between the lower bound and the constructions from \cite{ramkumar_access-redundancy_2023}. 
For the pair $(\nu, \lambda) = (1,1/2)$, the lower bound matches the ``trivial''  construction.
However, there is still a gap between the best known construction and the lower bound for $\nu > 1$.  

\begin{figure}[ht]    
    \centering
    \includegraphics[scale=0.64]{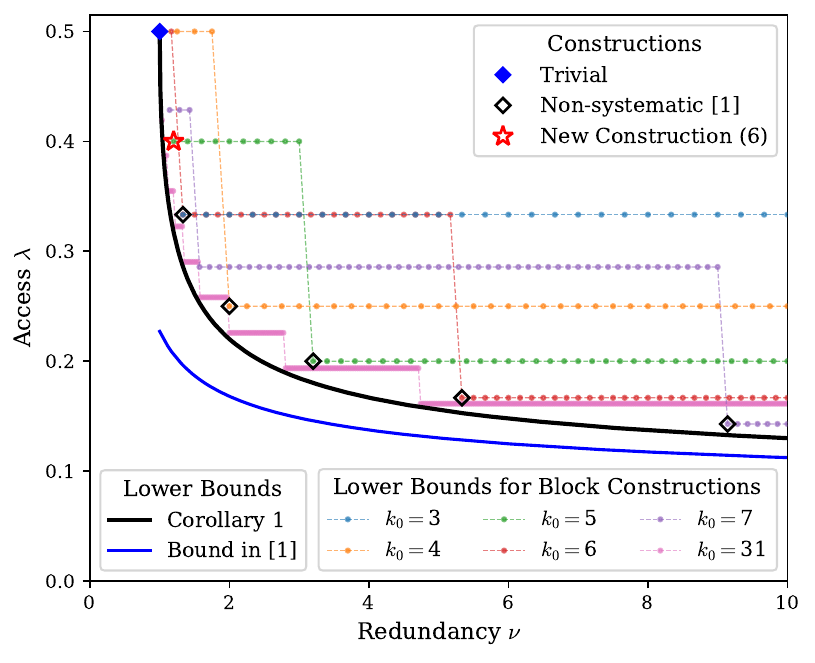}
    \caption{Numerical evaluation of the lower bounds and comparison to lower bound and the best constructions in \cite{ramkumar_access-redundancy_2023}. 
    The lower bounds for block constructions consist of dots only, the dashed lines are added for better readability.}
    \label{fig:compare}
\end{figure}

\subsection{Lower Bound for Block Constructions}
Recall from Section~\ref{sec:related} that the constructions in \cite{ramkumar_access-redundancy_2023} are 
based on dividing the whole length-$k$ vector $\mb x$ into $m$ blocks of size $k_0$, i.e. $\mb x=(\mb x_1,\dots,\mb x_m)$, and 
taking $\mb D=\mb I_m\otimes\mb M$ for some $\mb M\in\mbb R^{k_0\times n_0}$.
Moreover, for $\mb w=(\mb w_1,\dots,\mb w_m)$ with $\mb w_i\in\{\pm 1\}^{k_0}$, each $\mb w_i^{\mathrm T}\mb x_i$ is computed by querying at most $\ell_0$ entries of $\mb M^{\mathrm T}\mb x_i$.
In other words, each column of $\mb A$ can be written as $(\mb a_1,\dots, \mb a_m)$, where $\mb a_i\in \mbb R^{n_0}$ and $\lnorm \mb a_i\rnorm_0\leq \ell_0$.
Using this construction, we have $n=mn_0$ and $\ell=m\ell_0$. As $m$ goes to infinity, the pair $(\nu,\lambda)=(\frac{n_0}{k_0}, \frac{\ell_0}{k_0})$ is shown to be feasible. 
In the following, we refer to this type of construction as a \emph{block construction} with parameters $(k_0,n_0,\ell_0)$.

In \cite{ramkumar_access-redundancy_2023}, some specific block constructions are proposed. Here, we discuss whether there can be any possible improvements using this method.  

With similar techniques for the overall lower bound, we derive a lower bound for block constructions, which will show that in fact the non-systematic constructions of \cite{ramkumar_access-redundancy_2023} are optimal among block constructions. 
The intuition for the improved lower bound is as follows.  Using the block construction, the number of $\ell$-subsets of nodes that are used is at most $\binom{n_0}{\ell_0}^m$, which is much smaller than $\binom{n}{\ell} = \binom{mn_0}{m\ell_0}$, so the resulting lower bound should be strictly larger than the general one.

\begin{theorem}[Lower Bound for Block Constructions]\label{thm:block}
    A block construction of a $\{\pm 1\}$-protocol with parameters $(k_0,n_0,\ell_0)$ and $2\ell_0\leq n_0$ must satisfy
    $2\ell_0 + \log\binom{n_0}{\ell_0} - \log\binom{2\ell_0}{\ell_0} \geq k_0$.
\end{theorem}

\begin{proof}
Given the number of blocks $m\in\mbb N$, we have $k=mk_0, n=mn_0, \ell=m\ell_0$.
Let $\mscr A(n_0,\ell_0,m)\subseteq\binom{[n]}{\ell}$ be the collection of $\ell$-subsets such that there are $\ell_0$ elements in each block of size $n_0$. These are exactly the subsets of nodes that are possibly used in the block construction, and we know that $|\mscr A(n_0,\ell_0,m)|=\binom{n_0}{\ell_0}^m$.
Analogous to the covering designs in the proof of Theorem~\ref{thm:lower_bound}, we make the following claim.
\begin{claim}\label{cl:covdesign}
There exists a family $\mscr F\subseteq\mscr A(n_0,2\ell_0,m)$ 
with
\[|\mscr F|\leq \big[1+m\ln \textstyle\binom{2\ell_0}{\ell_0}\big]\displaystyle\frac{\binom{n_0}{\ell_0}^m}{\binom{2\ell_0}{\ell_0}^m}\]
such that $\forall\,\mcal S\in\mscr A(n_0,\ell_0,m)$, $\mcal S\subseteq\mcal T$ for some $\mcal T\in\mscr F$.
\end{claim}
Before we prove the claim, we see how it implies the theorem.
Using the same argument as in the proof of Theorem~\ref{thm:lower_bound}, we have 
\begin{align*}
    2^k &\leq \Bigg| \bigcup_{\mcal S \in \mscr A(n_0, \ell_0, m)} \hspace{-0.07cm}\Big\{ \{\pm1\}\text{-vectors in the span of } \mb D|_{\mcal S} \Big\} \Bigg| \\
    &\leq \sum_{\mcal T \in \mscr F} \Bigg| \bigcup_{\substack{\mcal S \subseteq \mcal T \\  \mcal{S} \in \mscr A(n_0, \ell_0, m)}}\hspace{-0.25cm}\Big\{ \{\pm1\}\text{-vectors in the span of } \mb D|_{\mcal S} \Big\} \Bigg| \\
    &\leq |\mscr F| 2^{2\ell},
\end{align*}
where the last line follows from \eqref{eq:key_observation}, as $|\mcal T| = 2\ell$ for all $\mcal T \in \mscr F$.  This implies that
$k\leq \log|\mscr F|+2\ell$. Since this should hold for arbitrarily large $m$, plugging in the expression from Claim~\ref{cl:covdesign}, we obtain the desired result. 

It remains to prove the Claim~\ref{cl:covdesign}.  We do so using the probabilistic method, similar to the proof Lemma~\ref{lemma:covering_design} in \cite{erdos_probabilistic_1974}. Let 
\[a = m\ln \textstyle\binom{2\ell_0}{\ell_0}\displaystyle\frac{\binom{n_0}{\ell_0}^m}{\binom{2\ell_0}{\ell_0}^m} \quad\text{and}\quad p = \frac{a}{\binom{n_0}{2\ell_0}^m}.\]
If $a \geq \binom{n_0}{2\ell_0}^m$, simply take $\mscr F=\mscr A(n_0,2\ell_0,m)$. Otherwise, randomly construct $\mscr F_p$ where each $\mcal T\in\mscr A(n_0,2\ell_0,m)$ is chosen independently with probability $p$. 
Thus $\E[|\mscr F_p|]=a$. 
Let $u(\mscr F_p) = \lbp \mcal S\in\mscr A(n_0,\ell_0,m) \mv \forall\,\mcal T\in\mscr F_p,\ \mcal S\not\subseteq\mcal T \rbp$ denote the uncovered sets. Then
\begin{align*}
    \Pr\{\mcal S\in u(\mscr F_p)\} &= (1-p)^{\binom{n_0-\ell_0}{\ell_0}^m}\leq e^{-p\binom{n_0-\ell_0}{\ell_0}^m},\\
    \E[|u(\mscr F_p)|]&\leq \textstyle\binom{n_0}{\ell_0}^m e^{-p\binom{n_0-\ell_0}{\ell_0}^m}.
\end{align*}
Let $v:\mscr A(n_0,\ell_0,m)\to\mscr A(n_0,2\ell_0,m)$ be a function such that $\mcal S\subseteq v(\mcal S)$ for all $\mcal S\in\mscr A(n_0,\ell_0,m)$. Then $\mscr F_p^* = \mscr F_p\cup \lbp v(\mcal S) \mv \mcal S\in u(\mscr F_p)\rbp$ covers all $\mcal S\in\mscr A(n_0,\ell_0,m)$ and 
\[\E[|\mscr F_p^*|]\leq a + \textstyle\binom{n_0}{\ell_0}^m e^{-p\binom{n_0-\ell_0}{\ell_0}^m} = \big[1+m\ln \textstyle\binom{2\ell_0}{\ell_0}\big]\displaystyle\frac{\binom{n_0}{\ell_0}^m}{\binom{2\ell_0}{\ell_0}^m}.\]
It suffices to pick some $\mscr F_p^*$ such that $|\mscr F_p^*|\leq \E[|\mscr F_p^*|]$.
This proves Claim~\ref{cl:covdesign} and hence the theorem.
\end{proof}

Numerical evaluations of the lower bound in Theorem~\ref{thm:block} for different block sizes are also provided in Figure~\ref{fig:compare}. 
The plot shows step behavior due to the restriction to block constructions where each block has the same structure $(k_0,n_0,\ell_0)$. In other words, for a fixed $k_0$, interpolation between the points will require different $n_0$ and $\ell_0$ for different blocks.
Notice that for block constructions with small block size $k_0$, it is impossible to have constructions that match the overall lower bound. This suggests that to close the gap between the lower bound and constructions, we have to either come up with a scheme that uses different approaches or further improve the lower bound.
To the best of our knowledge, the question of how to close this gap remains open.

Note that the non-systematic constructions in \cite{ramkumar_access-redundancy_2023} match a corner point on the lower bounds for the given block size.  
An interesting question is whether \emph{other} points on the lower bounds are also achievable.
One novel example of such a block construction is when $k_0=5$ with $n_0=6$ and
\begin{equation}\label{eq:new_construction}
    \mb M=\begin{pmatrix}
1 & -3 & 1 & 1 & 1 & 1\\
1 & 1 & -3 & 1 & 1 & 1\\
1 & 1 & 1 & -3 & 1 & 1\\
1 & 1 & 1 & 1 & -3 & 1\\
1 & 1 & 1 & 1 & 1 & -3
\end{pmatrix}
\end{equation}
It can be easily checked that each $\mb w_i\in\{\pm 1\}^5$ is in the span of at most $\ell_0=2$ columns of $\mb M$, resulting in the feasible pair $(6/5, 2/5)$. (We note that this feasible pair is the linear combination of two known feasible pairs $(1,1/2)$ and $(4/3, 1/3)$, so it does not improve the achievable region; but it does show that $(6/5, 2/5)$ can be achieved with a \emph{block construction}.)
In general, it remains open whether there exist block constructions that match the lower bounds everywhere.

%% file: v2/approx.tex
In this section, instead of exactly computing $f(\mb x)=\mb w^{\mathrm T}\mb x$, we slightly relax the requirement. 
Given $\varepsilon\in(0,1)$, for $\mb w\in\{\pm 1\}^k$,
now the goal is to output some $z$ such that $|\mb w^{\mathrm T}\mb x - z|^2\leq \varepsilon k \lnorm\mb x\rnorm_2^2$.
Without loss of generality, assume $\lnorm\mb x\rnorm_2=1$. Again we focus on linear schemes.
It suffices to design an encoding matrix $\mb D$ and a recovery matrix $\mb A$ so that each column of $\mb{W-DA}$ has squared $\ell_2$-norm at most $\varepsilon k$. To see this, recall that $\mb y=\mb D^{\mathrm T} \mb x$.
Given $\mb w\in\{\pm 1\}^k$, which is a column of $\mb W$, let $\mb a$ be the corresponding column of $\mb A$.  Then the scheme will output $z=\mb a^{\mathrm T}\mb y$, and we have
\[|\mb w^{\mathrm T}\mb x - z|^2 = |\mb w^{\mathrm T}\mb x - \mb a^{\mathrm T}\mb D^{\mathrm T} \mb x|^2 \leq \lnorm\mb {w - Da}\rnorm_2^2\lnorm\mb x\rnorm_2^2\leq \varepsilon k.\]
Hence in the following we focus on designing $\mb D$ and $\mb A$ to approximate $\mb W$ in this sense.
Similar as before, we call this an $\varepsilon$-approximation protocol with parameters $(k,n,\ell)$. The goal is to characterize the region of feasible pairs for $\varepsilon$-approximation protocols.

\subsection{Constructions}
\subsubsection{Covering Codes}
Using the same idea as in \cite{ramkumar_access-redundancy_2023}, we have the following proposition.
\begin{proposition}
Given a covering code $\mcal C$ with length $k_0$ and covering radius $r\geq 1$ over alphabet $\{\pm 1\}$, and an integer $b\in\{0,\dots, r-1\}$, then the pair $(\frac{k_0+\hat c}{k_0}, \frac{r-b+1}{k_0})$ is feasible with $\varepsilon\leq\frac{4b(k_0-r)}{k_0(k_0-r+b)}$. Recall that $\hat{c} = |\hat{\mcal C}|$, and $\hat{\mcal{C}} \subseteq \mcal C$ is a code with exactly one of $\pm \mb c$ for all $\mb c\in \mcal C$.  
Further, the pair $(\frac{\hat c}{k_0}, \frac{1}{k_0})$ is feasible with $\varepsilon\leq\frac{4r(k_0-r)}{k_0^2}$.
\end{proposition}
\begin{proof}
For $b\in\{0,\dots, r-1\}$, we construct a block construction following the scheme of~\cite[Theorem 2]{ramkumar_access-redundancy_2023}, which is also described in Section~\ref{subsubsec:construct}.
Suppose we have $m$ blocks, and $k=mk_0$.  Then for each block, instead of querying an additional $r$ nodes to correct all the errors, as is done in \cite{ramkumar_access-redundancy_2023}, we query at most $r-b$ additional nodes. 
Let $\bar{\mb c} \in \{\pm 1\}^k$ be the element of $\{ (\mb c_1, \mb c_2, \ldots, \mb c_m) \mid \mb c_i \in \mcal C \ \forall\, i \in [m]\}$ that is closest to $\mb w$. Let $\mcal I\subseteq[k]$ with $|\mcal I|\leq m(r-b)$ be a subset of the indices 
on which $\mb w$ and $\bar{\mb c}$ differ. Now we approximate $\mb w^{\mathrm T}\mb x$ with a linear combination of the $m$ nodes containing $\mb c_i^{\mathrm T}\mb x_i$ or $-\mb c_i^{\mathrm T}\mb x_i$, along with the nodes containing $x_j=\mb e_j^{\mathrm T}\mb x$ for $j\in\mcal I$.
Since $w_i\in\{\pm 1\}$, it makes sense to restrict the coefficients to be identical to each other within the first part and the second part, and 
$\varepsilon$ is obtained by solving $\varepsilon k\leq\max_{\mb w} \min_{\alpha,\beta}\big\lVert \alpha\bar{\mb c} + \beta \sum_{j\in\mcal I}\mb e_j -\mb w \big\lVert_2^2$. 
Note that the Hamming distance between $\mb w$ and $\bar{\mb c}$ is $mr$ in the worst case, solving the proposed optimization problem suggests taking $\beta=\alpha+1$ with $\alpha=\frac{k_0-r-b}{k_0-r+b}$. 
For the pair $(\frac{\hat c}{k_0}, \frac{1}{k_0})$, the scheme is constructed by not querying any additional nodes, thus no need to store each $x_j$. The corresponding $\varepsilon$ is calculated using the same arguments.
\end{proof}

\subsubsection{Discarding Blocks}
This method is again modified from \cite{ramkumar_access-redundancy_2023}. However, it applies to any block construction, not just those where the matrix $\mb M$ is found via covering codes.
\begin{proposition}
    Given a family of protocols using a block construction, with feasible pair $(\nu,\lambda)$, there is a family of $\varepsilon$-approximation protocols with feasible pair $\big((1-\varepsilon)\nu,(1-\varepsilon)\lambda\big)$.
\end{proposition}
\begin{proof}
    When the number of block is $m$, simply discard $\lfl\varepsilon m\rfl$ blocks.
\end{proof}

\subsubsection{K-SVD Algorithm}
Our final approach leverages the 
 K-SVD algorithm of~\cite{aharon_k-svd_2006}, which aims to construct overcomplete dictionaries for sparse representations.  This algorithm can be used to approximate $\min_{\mb{D,A}} \lnorm \mb{W-DA} \rnorm_F^2$
subject to each column of $\mb A$ having $\ell_0$-norm at most $\ell$. We run the K-SVD algorithm on our $\mb W$, for small values of $k, n, \ell$. The resulting $\varepsilon$ is computed as the worst-case approximation error among all $\mb w$, i.e. $\max_{\mb w}\lnorm\mb {w - Da}\rnorm_2^2$.  Then, using a block construction, we can generate an infinite family of protocols.  Our results are shown in Figure~\ref{fig:approx}. Notice that the performance does not seem as good as the other constructions; this may be because the objective function of the K-SVD algorithm can be viewed as the average error over $\mb w$, while $\varepsilon$ is determined by the worst case.

\subsection{Results}
Figure~\ref{fig:approx} shows the constructions using the three methods outlined above. Observe that using the second method, for $\varepsilon=0.1$, it is possible to exceed the lower bound for exact construction.

\begin{figure}[ht]    
    \centering
    \includegraphics[scale=0.64]{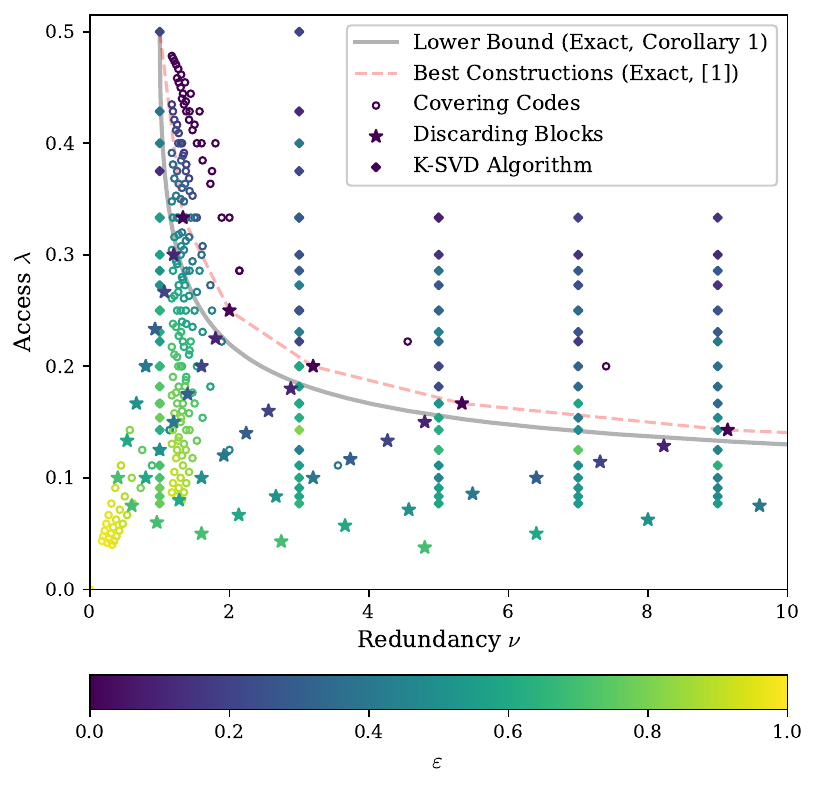}
    \caption{Our three approaches for approximate recovery. 
    }
    \label{fig:approx}
\end{figure}